\documentclass{article}
\usepackage{a4wide}
\usepackage{amsmath}
\usepackage{amssymb}
\usepackage{amsthm}

\usepackage{nameref}

\newtheorem{thm}{Theorem}
\newtheorem{prop}{Proposition}

\newtheorem{lemma}{Lemma}

\newcommand{\bbr}{{\mathbb R}}
\newcommand{\bbs}{{\mathbb S}}

\newcommand{\p}{\hat{p}}
\newcommand{\q}{\hat{q}}

\usepackage{authblk}

\begin{document}
\title{Late-time behaviour of the Einstein-Boltzmann system with a positive cosmological constant}

\author[1]{Ho Lee\footnote{holee@khu.ac.kr}}
\author[2]{Ernesto Nungesser\footnote{ernesto.nungesser@icmat.es}}

\affil[1]{Department of Mathematics and Research Institute for Basic Science, Kyung Hee University, Seoul, 02447, Republic of Korea}
\affil[2]{Instituto de Ciencias Matem\'{a}ticas (CSIC-UAM-UC3M-UCM), 28049 Madrid, Spain}

\maketitle

\begin{abstract}
In this paper we study the Einstein-Boltzmann system for Israel particles with a positive cosmological constant. We consider spatially homogeneous solutions of Bianchi types except IX and obtain future global existence and asymptotic behaviour of solutions to the Einstein-Boltzmann system. The result shows that the solutions converge to the de Sitter solution at late times.
\end{abstract}

\section{Introduction}
Einstein's equations are the basic equations of general relativity and can be used to study the time evolution of our Universe. They describe the interaction between geometry of spacetime and matter distribution in it so that a suitable matter model should be chosen and coupled to them. Many different types of matter models can be considered, for instance the perfect fluid is one of the most frequently used matter models in general relativity (see Section 3 of \cite{Rendall} for other matter models). In this paper we are interested in a kinetic matter model. In kinetic theory matter is considered as a collection of particles, and the Boltzmann equation describes the time evolution of particles interacting with each other through binary collisions. When the effect of collisions is negligible one may consider the Vlasov equation, and there are many results known for the Einstein-Vlasov system (see \cite{R} and the references of it). In this paper we are interested in the Einstein-Boltzmann system. The main interest of this paper is to study the late time behaviour of solutions to the Einstein-Boltzmann system with a positive cosmological constant.

In the presence of a positive cosmological constant it is expected that solutions to any Einstein-matter equations converge to the de Sitter solution at late times (at least for generic solutions). This is called the cosmic no hair conjecture, and in this paper we study this topic with the Einstein-Boltzmann system. We will consider spatially homogeneous spacetimes of Bianchi type. In this case Wald \cite{W} showed that solutions to the Einstein equations converge to the de Sitter solution provided that matter model chosen satisfies some energy conditions and the Bianchi model considered is not type IX. Hayoung Lee \cite{Lee} obtained the result for the Einstein-Vlasov system together with detailed asymptotic behaviour of the metric and matter terms. In the Boltzmann case we obtained similar results in \cite{LN2,LN3} where the FLRW and the Bianchi I cases were studied. In this paper we extend \cite{LN2,LN3} to all Bianchi cases except type IX. For the scattering kernel of the collision operator we consider the one for Israel particles. The main result of this paper is given in Theorem \ref{Thm}, and the proof is to use the energy method as in \cite{LN2}. 

This paper is organised as follows. In Section 2 we introduce the notations and terminology which will be assumed in this paper. In the second part of the section we introduce the Einstein-Boltzmann system for Israel particles with a positive cosmological constant having Bianchi symmetry. The main theorem will also be stated in this part. In Section 3 and 4 we study the Einstein equations and the Boltzmann equation, respectively, and show that solutions have certain desired properties. At the end of Section 4 we combine the results to prove the main theorem. In Section 5 we discuss the result of the present paper and related problems.

\section{The Einstein-Boltzmann system with Bianchi symmetry}
\subsection{Terminology}\label{Sec T}
In this paper we are interested in spatially homogeneous spacetimes of Bianchi type. Let $G$ be a $3$-dimensional Lie group and let $E_a$ with $a=1,2,3$ be a basis of the Lie algebra. We consider a Lorentz manifold $M=I(\subset \bbr)\times G$ with a metric $g$ of the following type:
\begin{align}
g=-dt^2+\chi_{ab}\xi^a\xi^b,
\end{align}
where $\xi^a$ are the $1$-forms dual to the $E_a$, and the $\chi_{ab}$ is a symmetric and positive-definite matrix, which only depends on $t$. For any spacetimes of Bianchi type we can assume the metric to be as described above by using a left-invariant frame (see \cite{R} for more details). By abuse of notation we assume that $E_\alpha$ with $\alpha=0,1,2,3$ is a left-invariant frame with $E_0=\partial/\partial t$. Throughout the paper Greek indices run from $0$ to $3$, while Latin indices run from $1$ to $3$, and the Einstein summation convention will be assumed. Indices are lowered and raised by the metric $g_{\alpha\beta}$ and $g^{\alpha\beta}$, respectively, where $g^{\alpha\beta}$ is the inverse of the matrix $g_{\alpha\beta}$, and the $g_{\alpha\beta}$ are the components of the metric $g$.

At each spacetime point $x\in M$ the tangent space $T_xM$ is spanned by the basis vectors $E_\alpha|_x$.
We take a coordinate system on $T_xM$ defined by
\[
p^\alpha E_\alpha|_x\mapsto (p^0,p^1,p^2,p^3).
\]
In this paper we are interested in spatially homogeneous solutions to the Boltzmann equation so that the $x$-dependence of the basis vectors may be ignored due to the use of the left-invariant frame. We express a momentum simply as $p^\alpha E_\alpha$, and use the above coordinates to write the distribution function as
\[
f=f(t,p).
\]
Here, momentum variables without indices denote three dimensional vectors with upper indices:
\[
p=(p^1,p^2,p^3),
\]
while for covariant vectors with lower indices we write
\[
p_*=(p_1,p_2,p_3),
\]
and the zeroth components are understood to be derived from the mass shell condition:
\[
p^0=\sqrt{1+p_ap^a},
\]
where we have assumed that all the particles have unit mass.
In order to derive a representation of the collision operator of the Boltzmann equation we shall introduce an orthonormal frame. Let $e_\mu=e_\mu^\alpha E_\alpha$ denote an orthonormal frame, i.e. $g_{\alpha\beta}e^\alpha_\mu e^\beta_\nu=\eta_{\mu\nu}$, where $\eta_{\mu\nu}$ denotes the Minkowski metric, and we take $e_0=E_0$ for simplicity. In an orthonormal frame a momentum is written as
\[
p^\alpha E_\alpha=\p^\mu e_\mu,
\]
where the hat indicates that the momentum has been written in an orthonormal frame. Note that
\[
p^\alpha=\p^\mu e_\mu^\alpha,\quad \p_\mu=\eta_{\mu\nu}\p^\nu=p_\alpha e^\alpha_\mu,
\]
where the Minkowski metric applies in an orthonormal frame.

In this paper we use two different types of multi-index notations for high order derivatives. Suppose that we are considering a partial derivative of order $m$. The first type uses an $m$-tuple of integers between $1$ and $3$. Let $A=(a_1,\cdots,a_m)$ be an $m$-tuple of integers with $a_i\in\{1,2,3\}$ and $i=1,\cdots,m$. We write
\[
\partial^A=\partial^{a_1}\cdots\partial^{a_m},
\]
where $\partial^a=\partial /\partial p_a$ is the partial derivative with respect to $p_a$ for $a\in\{1,2,3\}$. In this case the total order of differentiation is denoted by $|A|=m$. The second type uses a $3$-tuple of integers which are non-negative. Let $\mathcal{I}=(i,j,k)$ be a triple of non-negative integers, and we write
\[
\partial^{\mathcal I}=(\partial^1)^i(\partial^2)^j(\partial^3)^k,
\]
where $\partial^a$ for $a=1,2,3$ are understood as above. In this case the total order of differentiation is given by $|{\mathcal I}|=i+j+k=m$. Note that $\partial^{a_1}\cdots\partial^{a_m}$ is a permutation of $(\partial^1)^i(\partial^2)^j(\partial^3)^k$. For multi-indices of the first type we use capital letters from the beginning of the alphabet, while for the second type we use capital letters from the middle of the alphabet. They will also be distinguished by using different typefaces. For partial derivatives in an orthonormal frame, we use hats in a similar way, and the derivatives with respect to $p_a$ and $\p_a$ are related to each other as $\partial^a=e^a_b\hat{\partial}^b$. For a multi-index $A=(a_1,\cdots,a_m)$ we have
\[
\partial^A=e^A_B\hat{\partial}^B,
\]
where $e^A_B=e^{a_1}_{b_1}\cdots e^{a_m}_{b_m}$ for $B=(b_1,\cdots,b_m)$.

We also consider the usual $l^2$-norm: for a three-dimensional vector $v$, we define
\[
|v|=\sqrt{\sum_{i=1}^3(v^i)^2},
\]
and note that $|\p|^2=\eta_{ab}\p^a\p^b$. With this notation we define the weight function:
\[
\langle p_*\rangle=\sqrt{1+|p_*|^2},
\]
and note that it is different from $p^0$ in general. With this weight function we define the norm of a function $f=f(t,p_*)$ as follows: for a non-negative integer $N$,
\begin{align*}
\|f(t)\|^2_{k,N}=\sum_{|A|\leq N}\|\partial^A f(t)\|_k^2,\quad \|f(t)\|_k^2=\int_{\bbr^3}\langle p_*\rangle^{2k}e^{p^0(t)}|f(t,p_*)|^2dp_*,
\end{align*}
where $k$ is a positive real number. Note that the norm $\|\cdot\|_{k,N}$ can also be defined in terms of the second type of multi-indices.

\subsection{The Einstein-Boltzmann system with Bianchi symmetry}
In a spatially homogeneous setting the Einstein equations are given by
\begin{align}
\dot{\chi}_{ab}&=2k_{ab},\label{evolution1}\\
\dot{k}_{ab}&=2k^c_ak_{bc}-kk_{ab}-R_{ab}+S_{ab}+\frac{1}{2}(\rho-S)\chi_{ab}+\Lambda \chi_{ab}.\label{evolution2}
\end{align}
Here, the dot denotes the differentiation with respect to time, the $k_{ab}$ and $k=\chi^{ab}k_{ab}$ are the second fundamental form and the trace of it, respectively, $R_{ab}$ is the Ricci tensor of the metric $\chi_{ab}$, and $\Lambda>0$ is the cosmological constant. For the constraint equations we have
\begin{align}
R-k_{ab}k^{ab}+k^2&=2\rho+2\Lambda,\label{constraint1}\\
\nabla^ak_{ab}&=-T_{0b},\label{constraint2}
\end{align}
where $R$ is the Ricci scalar of the metric $\chi_{ab}$. We refer to the Chapter 25 of \cite{R} for more details. The quantities $S_{ab}$, $\rho$, $S$, and $T_{0b}$ are matter terms which are to be derived from the Boltzmann equation. Let a distribution function $f$ be given, then the stress energy tensor $T_{\alpha\beta}$ is defined by
\begin{align}
T_{\alpha\beta}=(\det\chi)^{\frac12}\int_{\bbr^3}f(t,p)\frac{p_\alpha p_\beta}{p^0}dp,\label{stress energy}
\end{align}
and the matter terms are given by $S_{ab}=T_{ab}$, $\rho=T_{00}$, and $S=\chi^{ab}S_{ab}$.

The Boltzmann equation is written as
\[
\frac{\partial f}{\partial t}-\frac{1}{p^0}\Gamma^a_{\alpha\beta}p^\alpha p^\beta\frac{\partial f}{\partial p^a}=Q(f,f).
\]
Let $\nabla$ be a connection compatible with the metric $g$, then the connection coefficients are defined by $\nabla_{E_\alpha}E_\beta=\Gamma^\gamma_{\alpha\beta}E_\gamma$ and can be expressed in terms of the structure constants $[E_a,E_b]=C^c_{ab}E_c$ with the Koszul formula:
\begin{align*}
2 \langle \nabla_V W, X \rangle = V \langle W, X \rangle + W \langle X, V \rangle - X \langle V,W \rangle - \langle V, [ W,X] \rangle + \langle W,[X,V] \rangle + \langle X, [ V, W] \rangle,
\end{align*}
where we used the notation $g(V,W)= \langle V,W \rangle$. If we write the distribution function as $f=f(t,p_*)$, then the left hand side of the Boltzmann equation reduces to
\begin{align}
\frac{\partial f}{\partial t}+\frac{1}{p^0}C_{ba}^cp_cp^b\frac{\partial f}{\partial p_a}=Q(f,f).\label{Vlasov}
\end{align}
To obtain an expression of the right hand side we use the orthonormal frame approach: we write the Boltzmann equation in an orthonormal frame to obtain an expression of the collision operator from special relativity, and revert to the original left-invariant frame through the transformation $\p\leftrightarrow p$. As a result we obtain
\begin{align}
Q(f,f)&=(\det \chi)^{-\frac12}\int_{\bbr^3}\int_{\bbs^2}v_M\sigma(h,\theta)\Big(f(p_*')f(q_*')-f(p_*)f(q_*)\Big)d\omega dq_*.\label{Q}
\end{align}
Here, the $v_M$ is the M{\o}ller velocity defined by
\[
v_M=\frac{h\sqrt{s}}{4p^0q^0},
\]
where $h$ and $s$ are the relative momentum and the total energy, respectively, defined by
\[
h=\sqrt{(p_\alpha-q_\alpha)(p^\alpha-q^\alpha)},\quad s=-(p_\alpha+q_\alpha)(p^\alpha+q^\alpha).
\]
Let us write $n^\alpha=p^\alpha+q^\alpha$ for simplicity. The post-collision momenta $p_k'$ and $q_k'$ are written as
\begin{align*}
\left(
\begin{array}{c}
p'^0\\
p'_k
\end{array}
\right)=
\left(
\begin{array}{c}
\displaystyle
p^0+2\bigg(-q^0\frac{n_a e^a_b\omega^b}{\sqrt{s}}+q_ae^a_b\omega^b+\frac{n_ae^a_b\omega^bn_c q^c}{\sqrt{s}(n^0+\sqrt{s})}\bigg)\frac{n_de^d_i\omega^i}{\sqrt{s}}\\
\displaystyle
p_k+2\bigg(-q^0\frac{n_ae^a_b\omega^b}{\sqrt{s}}+q_ae^a_b\omega^b+\frac{n_ae^a_b\omega^bn_cq^c}{\sqrt{s}(n^0+\sqrt{s})}\bigg)
\bigg(g_{ka}e^a_b\omega^b+\frac{n_ae^a_b\omega^bn_k}{\sqrt{s}(n^0+\sqrt{s})}\bigg)
\end{array}
\right),
\end{align*}
and $q_k'=n_k-p_k'$, where $\omega=(\omega^1,\omega^2,\omega^3)\in\bbs^2$, and recall that the $e^a_b$ are the components of an orthonormal frame. For a background on the representation of post-collision momenta we refer to the Appendix of \cite{LN2}. The quantity $\sigma(h,\theta)$ is the scattering kernel. In this paper we will consider the scattering kernel for Israel particles of the following type:
\[
\sigma(h,\theta)=\frac{4}{h(4+h^2)},
\]
so that the dependence on the scattering angle $\theta$ will be ignored. Finally, we combine the expressions \eqref{Vlasov} and \eqref{Q} together with the assumption on the scattering kernel to obtain
\begin{align}
\frac{\partial f}{\partial t}+\frac{1}{p^0}C_{ba}^cp_cp^b\frac{\partial f}{\partial p_a}=(\det\chi)^{-\frac12}\iint \frac{1}{p^0q^0\sqrt{s}}\Big(f(p_*')f(q_*')-f(p_*)f(q_*)\Big)d\omega dq_*,\label{Boltzmann}
\end{align}
and this equation will be referred to as the Boltzmann equation in the present paper.

To summarise the Einstein-Boltzmann system for Israel particles comprises the equations \eqref{evolution1}--\eqref{stress energy} and \eqref{Boltzmann}. The following is the main theorem of this paper. The proof will be presented in the following sections.

\begin{thm}\label{Thm}
Consider the Einstein-Boltzmann system for Israel particles with Bianchi symmetries except type IX and a positive cosmological constant $\Lambda$. Let $\chi_{ab}(t_0)$, $k_{ab}(t_0)$, and $f(t_0)$ be initial data of the Einstein-Boltzmann system satisfying the constraint equations, where the Hubble variable is given by $H(t_0)\leq (6/5)^{1/2}\gamma$ with $\gamma=(\Lambda/3)^{1/2}$ and the distribution function is given by $\|f(t_0)\|_{k+1/2,N}<\infty$ with $k>5$ and $N\geq 3$. Then, there exists $\varepsilon>0$ such that if $\|f(t_0)\|_{k+1/2,N}<\varepsilon$, then there exist unique global-in-time classical solutions $\chi_{ab}$, $k_{ab}$, and $f$ to the Einstein-Boltzmann system corresponding to the initial data. The spacetime is geodesically future complete, the distribution function $f$ is non-negative, and the solutions satisfy the following asymptotic behaviour:
\begin{align*}
&H=\gamma+O(e^{-2\gamma t}),\\
&\sigma^{ab}\sigma_{ab}=O(e^{-2\gamma t}),\\
&\chi_{ab}=e^{2\gamma t}\Big(G_{ab}+O(e^{-\gamma t})\Big),\\
&\chi^{ab}=e^{-2\gamma t}\Big(G^{ab}+O(e^{-\gamma t})\Big),
\end{align*}
where $G_{ab}$ and $G^{ab}$ are constant matrices, and the distribution function satisfies in an orthonormal frame
\[
f(t,\hat{p})\leq C \varepsilon (1+e^{2\gamma t} |\hat{p}|^2)^{-\frac12 k} e^{-\frac12 p^0},
\]
where $C$ is a positive constant.
\end{thm}

\section{Estimates for the Einstein part}\label{Sec E}
\subsection{Assumption on the distribution function}
To obtain the existence of solutions to the Einstein-Boltzmann system we can use the standard iteration method. In this section we assume that a suitable distribution function is given and show that the Einstein equations have solutions with desired asymptotic properties. To be precise we assume that the distribution function satisfies
\begin{align}
f(t,p_*)\leq C\langle p_*\rangle^{-k},\label{assumption f}
\end{align}
where $C$ is a positive constant independent of the metric and $k>5$ is a real number. The restriction on $k$ is required for the stress energy tensor to be well-defined. 

\subsection{Estimates for the Einstein part}
Given a distribution function $f$ satisfying \eqref{assumption f} global-in-time existence of solutions to the evolution equations \eqref{evolution1}--\eqref{evolution2} is easily obtained by following the proofs of \cite{Lee} and \cite{R94}. In the Vlasov case one can assume that a distribution function has compact support and show that matter terms are bounded at finite times by considering characteristics of the Vlasov equation. In the Boltzmann case we assume that a distribution function decays at infinity such as \eqref{assumption f} and estimate the stress energy tensor to obtain boundedness of matter terms. Note that the stress energy tensor \eqref{stress energy} can be written as
\[
T_{\alpha\beta}=(\det\chi)^{-\frac12}\int_{\bbr^3}f(t,p_*)\frac{p_\alpha p_\beta}{p^0}dp_*,
\]
and the assumption \eqref{assumption f} ensures that $T_{\alpha\beta}$ is bounded. The rest of arguments are the same as in \cite{Lee} and \cite{R94}, and we obtain global-in-time solutions to the equations \eqref{evolution1}--\eqref{evolution2} for a given distribution function. 

Asymptotic behaviour of solutions to the equations \eqref{evolution1}--\eqref{evolution2} is also similarly obtained as in \cite{Lee}. We only briefly sketch the estimates for the asymptotic behaviour. We first note that the sign of the second fundamental form used in \cite{Lee} is different from the one in \cite{R}, and we will follow the sign convention of \cite{R}. We define the Hubble variable $H$ as
\[
H=\frac13 k,
\]
where $k=\chi^{ab}k_{ab}$, and decompose the second fundamental form as
\[
k_{ab}=\sigma_{ab}+H\chi_{ab}.
\]
The quantity $\sigma_{ab}$ is called the shear tensor, and we rewrite the constraint equation \eqref{constraint1} as
\[
\frac{R}{2}-\frac{\sigma_{ab}\sigma^{ab}}{2}-\rho=-3H^2+\Lambda.
\]
Since we are considering the spacetimes which are not the Bianchi type IX, the Ricci scalar $R$ is non-positive (see \cite{Lee,W} for more details), and we conclude that
\[
3H^2\geq\Lambda,
\]
which is equivalent to $H\geq\gamma=(\Lambda/3)^{1/2}$. We now use the evolution equations \eqref{evolution1}--\eqref{evolution2} to derive a differential inequality for $H$ such that $\dot{H}\leq -H^2+\Lambda/3$, and we obtain
\[
0\leq H-\gamma\leq \frac{2\gamma(H(t_0)-\gamma)}{H(t_0)+\gamma}e^{-2\gamma(t-t_0)},
\]
where $t_0$ is an initial time. To estimate the shear scalar $\sigma_{ab}\sigma^{ab}$ we consider the above constraint equation again such that
\begin{align}
\frac{\sigma_{ab}\sigma^{ab}}{2}\leq 3H^2-\Lambda= 3(H+\gamma)(H-\gamma),\label{inequality shear}
\end{align}
and conclude that
\[
0\leq \sigma_{ab}\sigma^{ab}\leq Ce^{-2\gamma(t-t_0)},
\]
where the constant $C$ depends on the initial data. We now estimate the metric $\chi$. Let us consider a scaled quantity $\bar{\chi}_{ab}=e^{-2\gamma t}\chi_{ab}$ and derive an equation
\[
\dot{\bar{\chi}}_{ab}=2(H-\gamma)\bar{\chi}_{ab}+2e^{-2\gamma t}\sigma_{ab}.
\]
This implies that $|\chi_{ab}|\leq Ce^{2\gamma t}$, which in turn implies $|\sigma_{ab}|\leq Ce^{\gamma t}$ (see \cite{Lee} for more details), hence the right hand side of the above differential equation is integrable. We may now define a constant matrix $G_{ab}$ such that
\[
G_{ab}=\bar{\chi}_{ab}(t_0)+2\int_{t_0}^\infty (H(s)-\gamma)\bar{\chi}_{ab}(s)+e^{-2\gamma s}\sigma_{ab}(s)ds,
\]
and conclude that
\[
\chi_{ab}=e^{2\gamma t}\Big(G_{ab}+O(e^{-\gamma t})\Big).
\]
Asymptotic behaviour of $\chi^{ab}$ is similarly obtained, and the results are summarized as follows.

\begin{prop}\label{Prop E}
Let a distribution function $f$ be given and satisfy \eqref{assumption f}. Suppose that $\chi_{ab}(t_0)$ and $k_{ab}(t_0)$ are initial data of the evolution equations \eqref{evolution1}--\eqref{evolution2}, which satisfy the constraint equations \eqref{constraint1}--\eqref{constraint2}. Then, the evolution equations \eqref{evolution1}--\eqref{evolution2} admit global-in-time solutions with the following asymptotic behaviour:
\begin{align*}
&H=\gamma+O(e^{-2\gamma t}),\\
&\sigma_{ab}\sigma^{ab}=O(e^{-2\gamma t}),\\
&\chi_{ab}=e^{2\gamma t}\Big(G_{ab}+O(e^{-\gamma t})\Big),\\
&\chi^{ab}=e^{-2\gamma t}\Big(G^{ab}+O(e^{-\gamma t})\Big),
\end{align*}
where $\gamma=(\Lambda/3)^{1/2}$, and $G_{ab}$ and $G^{ab}$ are constant matrices.
\end{prop}

We remark that the determinant of the $3$-metric $\chi$ is estimated as
\[
\det \chi=O(e^{6\gamma t}).
\]
Moreover, the components of the orthonormal frame introduced in Section \ref{Sec T} have the following asymptotic behaviour:
\[
e^a_b=O(e^{-\gamma t}),\quad (e^{-1})^a_b=O(e^{\gamma t}),
\]
where $e^{-1}$ is the inverse of the $3\times 3$ matrix $e$. An explicit formula of an orthonormal frame is given in \cite{LN1}. For later use we consider the following quantity:
\[
F=\frac{\sigma_{ab}\sigma^{ab}}{4H^2}.
\]
If we assume that the Hubble variable is initially bounded such that $H(t_0)\leq(6/5)^{1/2}\gamma$, then by the inequality \eqref{inequality shear} we have
\begin{align}\label{boundedness F}
F\leq \frac{3H^2-3\gamma^2}{2H^2}\leq\frac32\bigg(1-\frac{\gamma^2}{H^2(t_0)}\bigg)\leq\frac14,
\end{align}
where we used the fact that $H$ is non-increasing in time. We remark that the boundedness of $F$ will be used to have monotonicity of $p^0$ in later sections. 

\section{Proof of the main theorem}\label{Sec M}
In this part we prove the main theorem. In Section \ref{Sec AM} we make suitable assumptions on the metric so that the Boltzmann equation admits global-in-time solutions. The global existence is proved in Section \ref{Sec B} under the assumptions. Note that the metric given in Proposition \ref{Prop E} satisfies the assumptions in Section \ref{Sec AM}, and the distribution function given in Proposition \ref{Prop B} satisfies the assumption \eqref{assumption f}. This shows that the iteration for the Einstein-Boltzmann system is well-defined and admits global-in-time solutions with the desired asymptotic behaviour. This will be discussed in Section \ref{Sec EB}.

\subsection{Assumptions on the metric}\label{Sec AM}
Let $\chi_{ab}$ be a given spatial metric, and we assume that the metric satisfies the following properties. Let us write $\gamma=(\Lambda/3)^{1/2}$. For any $t\geq t_0$ and $p_*\in\bbr^3$, the metric satisfies
\[
\frac{1}{C}e^{-2\gamma t}|p_*|^2\leq \chi^{ab}p_ap_b\leq C e^{-2\gamma t}|p_*|^2
\]
for some $C>0$. The components of the metric satisfy $|\chi_{ab}|\leq Ce^{2\gamma t}$ and $|\chi^{ab}|\leq Ce^{-2\gamma t}$, and spatial components of the corresponding orthonormal frame satisfy $|e^a_b|\leq Ce^{-\gamma t}$ and $|(e^{-1})^a_b|\leq Ce^{\gamma t}$. The Hubble variable $H$ is bounded, and the scaled shear scalar satisfies $F\leq 1/4$ for all $t\geq t_0$. 

\subsection{Estimates for the Boltzmann part}\label{Sec B}
In this part we study the existence of solutions to the Boltzmann equation under the assumptions of Section \ref{Sec AM}. We first collect several basic lemmas and then provide the main estimates which are given in Lemma \ref{Lem f} and \ref{Lem partial f}. 

\begin{lemma}\label{Lem p}
Let $p^\alpha$ be a momentum satisfying the mass shell condition. Then, we have
\begin{align*}
&\frac{\partial p^0}{\partial p_a}=\frac{p^a}{p^0},\\
&\frac{\partial p^0}{\partial t}=-\frac{(\sigma^{ab}+H\chi^{ab})p_ap_b}{p^0},\\
&\frac{\partial p^b}{\partial p_a}=\chi^{ab},\allowdisplaybreaks\\
&\frac{\partial \langle p_*\rangle^{2k}}{\partial p_a}=2k\eta^{ab}p_b\langle p_*\rangle^{2k-2},
\end{align*}
where $\eta^{ab}$ is the Minkowski metric and $k$ is a real number.
\end{lemma}
\begin{proof}
The lemma is proved by simple calculations, and we skip the proof.
\end{proof}

For the proofs of Lemma \ref{Lem pp'q'}, \ref{Lem 1/p^0}, and \ref{Lem partial p'} below we refer to \cite{LN3} where the Bianchi I case is studied. Since the asymptotic behaviour of the spatial metric assumed in Section \ref{Sec AM} is the same with the one considered in the Bianchi I case, the proof given in \cite{LN3} still applies to Lemma \ref{Lem pp'q'}. To prove the results of Lemma \ref{Lem 1/p^0} and \ref{Lem partial p'} one can use the orthonormal frame approach, hence they are exactly the same with the ones given in \cite{LN3}. We also refer to \cite{LN2} for more details.

\begin{lemma}\label{Lem pp'q'}
Let $p_*'$ and $q_*'$ be post-collision momenta for given $p_*$ and $q_*$. Then, we have
\[
\langle p_*\rangle\leq C\langle p_*'\rangle\langle q_*'\rangle,
\]
where $C$ is a positive constant.
\end{lemma}
\begin{proof}
We refer to \cite{LN3} for the proof.
\end{proof}

\begin{lemma}\label{Lem 1/p^0}
For a multi-index $A$, there exist polynomials ${\mathcal P}$ and ${\mathcal P}_i$ such that
\begin{align*}
&\partial^A \bigg[\frac{1}{p^0}\bigg]=\frac{e^A_B}{(p^0)^{|A|+1}}{\mathcal P}\bigg(\frac{\p}{p^0}\bigg),\\
&\partial^A\bigg[\frac{1}{\sqrt{s}}\bigg]=\frac{e^A_C}{\sqrt{s}}\sum_{i=0}^{|A|}\bigg(\frac{q^0}{s}\bigg)^i \bigg(\frac{1}{p^0}\bigg)^{|A|-i}{\mathcal P}_i\bigg(\frac{\p}{p^0},\frac{\q}{q^0}\bigg),
\end{align*}
where the multi-indices $B$ and $C$ are summed with the polynomials.
\end{lemma}
\begin{proof}
We refer to \cite{LN3} for the proof.
\end{proof}

\begin{lemma}\label{Lem partial p'}
Let $p_*'$ and $q_*'$ be post-collision momenta for given $p_*$ and $q_*$. For a multi-index $A\neq 0$, we have
\[
|\partial^Ap'_*|+|\partial^Aq'_*|\leq C\max_{a,b}|(e^{-1})^a_b|(\max_{c,d}|e^c_d|)^{|A|}(q^0)^{|A|+4},
\]
where $C$ is a positive constant.
\end{lemma}
\begin{proof}
We refer to \cite{LN3} for the proof.
\end{proof}

\begin{lemma}\label{Lem pp}
Let $\mathcal{I}$ be a multi-index with $|\mathcal{I}|=m\geq 1$. Then we have
\[
\bigg|\partial^{\mathcal I}\bigg[\frac{p^ap_b}{p^0}\bigg]\bigg|\leq Ce^{-m\gamma t},
\]
where $C$ is a positive constant.
\end{lemma}
\begin{proof}
We use the orthonormal frame approach. Let us write $p^a=e^a_c \p^c$ and $p_b=(e^{-1})^d_b \p_d$, and $A$ be a multi-index of the first type satisfying $\partial^A=\partial^{\mathcal I}$. The quantity on the left hand side is written as
\[
\partial^{\mathcal I}\bigg[\frac{p^ap_b}{p^0}\bigg]=e^a_c(e^{-1})^d_b\partial^A\bigg[\frac{\p^a\p_b}{p^0}\bigg]=e^a_c(e^{-1})^d_be^A_B\hat{\partial}^B\bigg[\frac{\p^a\p_b}{p^0}\bigg],
\]
where $B$ is a multi-index such that $|B|=|A|$. For each $B$ let $\mathcal{J}$ be the multi-index of the second type satisfying $\hat{\partial}^{\mathcal J}=\hat{\partial}^B$, and apply the general Leibniz rule to have
\[
\hat{\partial}^B\bigg[\frac{\p^a\p_b}{p^0}\bigg]=\hat{\partial}^{\mathcal{J}}\bigg[\frac{\p^a\p_b}{p^0}\bigg]=\sum_{\substack{\mathcal{J}=\mathcal{K}+\mathcal{L}\\ |\mathcal{K}|\leq 2}}\binom{\mathcal J}{\mathcal K}\hat{\partial}^{\mathcal K}\Big[\p^a\p_b\Big]\hat{\partial}^{\mathcal L}\bigg[\frac{1}{p^0}\bigg].
\]
The quantity $\hat{\partial}^{\mathcal K}\Big[\p^a\p_b\Big]$ for $|{\mathcal K}|\leq 2$ is estimated as follows:
\[
\Big|\hat{\partial}^{\mathcal K}\Big[\p^a\p_b\Big]\Big|\leq C(p^0)^{2-|{\mathcal K}|},
\]
where we used the fact that $\p^a$ and $\p_b$ are bounded by $p^0$. The estimate of high order derivatives of $1/p^0$ is given in Lemma \ref{Lem 1/p^0}, i.e.
\[
\bigg|\hat{\partial}^{\mathcal L}\bigg[\frac{1}{p^0}\bigg]\bigg|\leq \frac{C}{(p^0)^{|\mathcal{L}|+1}}.
\]
Combine the two estimates to obtain
\[
\bigg|\hat{\partial}^B\bigg[\frac{\p^a\p_b}{p^0}\bigg]\bigg|\leq C (p^0)^{1-|B|}.
\]
Note that $|\mathcal{I}|=|A|=|B|$ and $|\mathcal{I}|\geq 1$, hence the above quantity is bounded. We conclude that
\[
\bigg|\partial^{\mathcal I}\bigg[\frac{p^ap_b}{p^0}\bigg]\bigg|\leq C(\max_{a,b}|e^a_b|)^{1+|\mathcal{I}|}(\max_{c,d}|(e^{-1})^c_d|),
\]
and use the assumptions in Section \ref{Sec AM} to obtain the desired result.
\end{proof}

\begin{lemma}\label{Lem f}
Let $f$ be a classical solution to the Boltzmann equation \eqref{Boltzmann} with initial data $f(t_0)$. If $\|f(t_0)\|_k$ is sufficienly small, then we have
\[
\sup_{t\geq t_0}\|f(t)\|_k \leq C\|f(t_0)\|_k,
\]
where $C$ is a positive constant.
\end{lemma}
\begin{proof}
Let us first consider the left hand side of the Boltzmann equation \eqref{Boltzmann}. Multiplying the left hand side by $\langle p_*\rangle^{2k}e^{p^0}f$ and integrating it with respect to $p_*$ we obtain
\begin{align*}
&\int\langle p_*\rangle^{2k}e^{p^0}f\bigg(\frac{\partial f}{\partial t}+\frac{1}{p^0}C^c_{ba}p^bp_c\frac{\partial f}{\partial p_a}\bigg)dp_*\\
&=\frac12\int\langle p_*\rangle^{2k}e^{p^0}\frac{\partial (f^2)}{\partial t}dp_*+\frac12\int\langle p_*\rangle^{2k}e^{p^0}\frac{1}{p^0}C^c_{ba}p^bp_c\frac{\partial (f^2)}{\partial p_a}dp_*.
\end{align*}
The first quantity on the right side above is written as
\begin{align*}
\frac12\int\langle p_*\rangle^{2k}e^{p^0}\frac{\partial (f^2)}{\partial t}dp_*
=\frac12\frac{d}{dt}\|f(t)\|_k^2+\frac12\int\langle p_*\rangle^{2k}e^{p^0}\frac{(\sigma^{ab}+H\chi^{ab})p_ap_b}{p^0} f^2 dp_*,
\end{align*}
where we used Lemma \ref{Lem p}, and note that the second quantity above is non-negative since
\begin{align}
|\sigma^{ab}p_ap_b|\leq (\sigma^{ab}\sigma_{ab})^{\frac12}(\chi^{cd}p_cp_d)=2HF^{\frac12}(\chi^{ab}p_ap_b)\leq H(\chi^{ab}p_ap_b).\label{est sigma}
\end{align}
We also have
\begin{align*}
&\frac12\int\langle p_*\rangle^{2k}e^{p^0}\frac{1}{p^0}C^c_{ba}p^bp_c\frac{\partial (f^2)}{\partial p_a}dp_*\\
&=-\frac12\int\Big(2k\eta^{ad}p_d\langle p_*\rangle^{2k-2}\Big)e^{p^0}\frac{1}{p^0}C^c_{ba}p^bp_cf^2 dp_*
-\frac12\int\langle p_*\rangle^{2k}e^{p^0}\frac{p^a}{(p^0)^2}C^c_{ba}p^bp_c f^2 dp_*\allowdisplaybreaks\\
&\quad +\frac12\int\langle p_*\rangle^{2k}e^{p^0}\frac{p^a}{(p^0)^3}C^c_{ba}p^bp_c f^2 dp_*
-\frac12\int\langle p_*\rangle^{2k}e^{p^0}\frac{1}{p^0}C^c_{ba}\chi^{ab}p_c f^2 dp_*\\
&\quad -\frac12\int\langle p_*\rangle^{2k}e^{p^0}\frac{1}{p^0}C^c_{ba}p^b\delta^a_c f^2 dp_*\allowdisplaybreaks\\
&=-\frac12\int\Big(2k\eta^{ad}p_d\langle p_*\rangle^{2k-2}\Big)e^{p^0}\frac{1}{p^0}C^c_{ba}p^bp_cf^2 dp_*
-\frac12\int\langle p_*\rangle^{2k}e^{p^0}\frac{1}{p^0}C^a_{ba}p^b f^2 dp_*,
\end{align*}
where we used the antisymmetry property of structure constants. Since $p_a$ is bounded by $\langle p_*\rangle$, and $p^a$ can be written as $e^a_b\p^b$, which is bounded by $Cp^0e^{-\gamma t}$, the above two integrals are estimated as follows:
\[
\bigg|\int\Big(2k\eta^{ad}p_d\langle p_*\rangle^{2k-2}\Big)e^{p^0}\frac{1}{p^0}C^c_{ba}p^bp_cf^2 dp_*\bigg|
+\bigg|\int\langle p_*\rangle^{2k}e^{p^0}\frac{1}{p^0}C^a_{ba}p^b f^2 dp_*\bigg|\leq Ce^{-\gamma t}\|f(t)\|_k^2.
\]
We conclude that the left hand side of the Boltzmann equation \eqref{Boltzmann} is estimated as
\begin{align}\label{est f left}
\int\langle p_*\rangle^{2k}e^{p^0}f\bigg(\frac{\partial f}{\partial t}+\frac{1}{p^0}C^c_{ba}p^bp_c\frac{\partial f}{\partial p_a}\bigg)dp_*\geq \frac12\frac{d}{dt}\|f(t)\|_k^2-Ce^{-\gamma t}\|f(t)\|^2_k.
\end{align}
The estimate of the right hand side of the Boltzmann equation \eqref{Boltzmann} is the same as the estimate given in \cite{LN3} where the Bianchi I case is studied. We also refer to \cite{LN2} for more details. Following the estimates (9) of \cite{LN3} we obtain
\begin{align}\label{est f right}
\bigg|\int \langle p_*\rangle^{2k} e^{p^0}f(t,p_*)Q(f,f)(t,p_*)dp_*\bigg|\leq C (\det\chi)^{-\frac14}\|f(t)\|_k^3.
\end{align}
Combine the estimates \eqref{est f left} and \eqref{est f right} to obtain
\begin{align}\label{est f}
\frac12\frac{d}{dt}\|f(t)\|_k^2\leq Ce^{-\gamma t}\|f(t)\|_k^2+Ce^{-\frac32\gamma t}\|f(t)\|_k^3,
\end{align}
where we used the assumptions of Section \ref{Sec AM}. 

To estimate the differential inequality \eqref{est f}, let us write $u(t)=\|f(t)\|_k^2$ for simplicity so that the inequality is written as
\[
-\frac{d}{dt}\bigg[\frac{1}{\sqrt{u}}\bigg]\leq Ce^{-\gamma t}\frac{1}{\sqrt{u}}+Ce^{-\frac32\gamma t},
\]
as long as $u$ does not vanish. Multiplying a suitable integrating factor we have
\begin{align*}
\frac{1}{\sqrt{u(t_0)}} &\leq \frac{\exp \int_{t_0}^t Ce^{-\gamma s}ds}{\sqrt{u(t)}}+C\int_{t_0}^t e^{-\frac32\gamma s}\exp\int_{t_0}^s Ce^{-\gamma \tau}d\tau ds\leq \frac{c_1}{\sqrt{u(t)}}+c_2,
\end{align*}
where $c_1$ and $c_2$ are finite numbers. We conclude that if $u(t_0)$ is small such as $u(t_0)< c_2^{-2}$, then $u(t)$ is bounded by
\[
u(t)\leq \bigg(\frac{c_1\sqrt{u(t_0)}}{1-c_2\sqrt{u(t_0)}}\bigg)^2,
\]
and this proves the desired result.
\end{proof}

\begin{lemma}\label{Lem partial f}
Let $f$ be a classical solution to the Boltzmann equation \eqref{Boltzmann} with initial data $f(t_0)$. If $\|f(t_0)\|_{k,N}$ is sufficiently small, then we have
\[
\sup_{t\geq t_0} \|f(t)\|_{k,N}\leq C\|f(t_0)\|_{k,N},
\]
where $C$ is a positive constant.
\end{lemma}
\begin{proof}
Let ${\mathcal I}\neq 0$ be a multi-index of the second type, and take the derivative $\partial^{\mathcal I}$ on the left hand side of the Boltzmann equation \eqref{Boltzmann} to have
\begin{align*}
&\partial^{\mathcal I}\bigg[\frac{\partial f}{\partial t}+\frac{1}{p^0}C^c_{ba}p^bp_c\frac{\partial f}{\partial p_a}\bigg]\\
&=\frac{\partial (\partial^{\mathcal I}f)}{\partial t}+\frac{1}{p^0}C^c_{ba}p^bp_c\frac{\partial (\partial^{\mathcal I}f)}{\partial p_a}
+\sum_{\substack{\mathcal{ I=J+K}\\ {|\mathcal J|\geq 1}}}\binom{\mathcal I}{\mathcal J}\partial^{\mathcal J}\bigg[\frac{1}{p^0}C^c_{ba}p^bp_c\bigg]\frac{\partial (\partial^{\mathcal K}f)}{\partial p_a}.
\end{align*}
The estimates of the first and the second quantities on the right side above are the same as the previous lemma. Multiplying the above by $\langle p_*\rangle^{2k}e^{p^0}\partial^{\mathcal I}f$ and integrating it with respect to $p_*$ we obtain from the first and the second quantities
\begin{align}
&\frac12\frac{d}{dt}\|\partial^{\mathcal I}f(t)\|_k^2+\frac12\int\langle p_*\rangle^{2k}e^{p^0}\frac{(\sigma^{ab}+H\chi^{ab})p_ap_b}{p^0} (\partial^{\mathcal I}f)^2 dp_*\nonumber\\
&\quad -\frac12\int\Big(2k\eta^{ad}p_d\langle p_*\rangle^{2k-2}\Big)e^{p^0}\frac{1}{p^0}C^c_{ba}p^bp_c(\partial^{\mathcal I}f)^2 dp_*
-\frac12\int\langle p_*\rangle^{2k}e^{p^0}\frac{1}{p^0}C^a_{ba}p^b(\partial^{\mathcal I} f)^2 dp_*\nonumber\\
&\geq \frac12\frac{d}{dt}\|\partial^{\mathcal I}f(t)\|_k^2-Ce^{-\gamma t}\|\partial^{\mathcal I}f(t)\|_k^2.\label{est partial f left 1}
\end{align}
The third quantity is estimated as
\begin{align}
&\bigg|\sum_{\substack{\mathcal{ I=J+K}\\ {|\mathcal J|\geq 1}}}\binom{\mathcal I}{\mathcal J} \int\langle p_*\rangle^{2k}e^{p^0}(\partial^{\mathcal I}f) \partial^{\mathcal J}\bigg[\frac{1}{p^0}C^c_{ba}p^bp_c\bigg]\frac{\partial (\partial^{\mathcal K}f)}{\partial p_a}dp_*\bigg|\nonumber\\
&\leq Ce^{-\gamma t} \bigg(\|\partial^{\mathcal I}f(t)\|_k^2 + \sum_{1\leq|{\mathcal L}|\leq |{\mathcal I}|}\|\partial^{\mathcal L}f(t)\|^2_k\bigg)
\leq Ce^{-\gamma t} \|f(t)\|^2_{k,|{\mathcal I}|},\label{est partial f left 2}
\end{align}
where we used Lemma \ref{Lem pp} and the fact that $|{\mathcal J}|\geq 1$ and $|{\mathcal K}|\leq |{\mathcal I}|-1$. 

We now use Lemma \ref{Lem pp'q'}, \ref{Lem 1/p^0}, and \ref{Lem partial p'} to estimate the right hand side of the Boltzmann equation. Following the estimates given in the proof of Lemma 10 of \cite{LN3} we obtain
\begin{align}
\int \langle p_*\rangle^{2k}e^{p^0}(\partial^{\mathcal I} f)\partial^{\mathcal I} Q(f,f)dp_*\leq C(\det \chi)^{-\frac14}\bigg(\sum_{|{\mathcal L}|\leq |{\mathcal I}|}\| \partial^{\mathcal L}f(t)\|_k\bigg)^3.\label{est partial f right}
\end{align}
We combine the estimates \eqref{est partial f left 1}, \eqref{est partial f left 2}, and \eqref{est partial f right} to obtain
\begin{align*}
&\frac12\frac{d}{dt}\|\partial^{\mathcal I}f(t)\|^2_k
\leq Ce^{-\gamma t}\|\partial^{\mathcal I}f(t)\|^2_k
+Ce^{-\gamma t}\|f(t)\|^2_{k,|{\mathcal I}|}
+Ce^{-\frac32\gamma t}\|f(t)\|^3_{k,|{\mathcal I}|},
\end{align*}
where we use the assumptions of Section \ref{Sec AM}. Note that the case ${\mathcal I}=0$ is given by the previous lemma. Collecting all the possible $|{\mathcal I}|\leq N$ we obtain
\begin{align}\label{est partial f}
\frac12\frac{d}{dt}\|f(t)\|^2_{k,N}\leq Ce^{-\gamma t}\|f(t)\|^2_{k,N}+Ce^{-\frac32\gamma t}\|f(t)\|^3_{k,N},
\end{align}
which is the same differential inequality with \eqref{est f}. Consequently, by the same argument as the previous lemma we obtain the desired result for small initial data.
\end{proof}

\begin{prop}\label{Prop B}
Let a spatial metric $\chi_{ab}$ be given and satisfy the assumptions of Section \ref{Sec AM}. Then, there exists $\varepsilon>0$ such that if initial data satisfies $\|f(t_0)\|_{k+1/2,N}<\varepsilon$ for $N\geq 3$, then there exists a non-negative global-in-time solution to the Boltzmann equation \eqref{Boltzmann} satisfying
\[
\sup_{t\in[t_0,\infty)}\|f(t)\|_{k,N}\leq C\varepsilon.
\]
\end{prop}
\begin{proof}
Local-in-time existence is proved by a standard iteration, and Lemma \ref{Lem partial f} shows that the solution is bounded globally in time with respect to the norm $\|\cdot\|_{k,N}$. This proves the global-in-time existence and the boundedness of solutions. We remark that the standard iteration will be given by
\[
\frac{\partial f_{n+1}}{\partial t}+\frac{1}{p^0}C_{ba}^cp_cp^b\frac{\partial f_{n+1}}{\partial p_a}=(\det\chi)^{-\frac12}\iint \frac{1}{p^0q^0\sqrt{s}}\Big(f_n(p_*')f_n(q_*')-f_n(p_*)f_{n+1}(q_*)\Big)d\omega dq_*,
\]
and this ensures the non-negativity of the solutions. The second term on the left hand side of the inequality \eqref{est partial f left 1} is non-negative by \eqref{est sigma} and has been ignored to obtain the estimate \eqref{est partial f}. It can be estimated again by using \eqref{est sigma} as follows:
\begin{align*}
\bigg|\int\langle p_*\rangle^{2k}e^{p^0}\frac{(\sigma^{ab}+H\chi^{ab})p_ap_b}{p^0} (\partial^{\mathcal I}f)^2 dp_*\bigg|\leq 2H\int \langle p_*\rangle^{2k}e^{p^0}\frac{\chi^{ab}p_ap_b}{p^0} (\partial^{\mathcal I}f)^2 dp_*.
\end{align*}
Note that $\chi^{ab}$ is bounded by $Ce^{-2\gamma t}$. Moreover, $p_a$ is bounded by $\langle p_*\rangle$, and $p_b/p^0=(e^{-1})^c_b\p_c/p^0$ is bounded by $Ce^{\gamma t}$, hence we have
\begin{align*}
2H\int \langle p_*\rangle^{2k}e^{p^0}\frac{\chi^{ab}p_ap_b}{p^0} (\partial^{\mathcal I}f)^2 dp_*
\leq Ce^{-\gamma t}\|\partial^{\mathcal I}f(t)\|^2_{k+\frac12}.
\end{align*}
Following the proof of Lemma \ref{Lem partial f} we obtain
\begin{align*}
\bigg|\frac12\frac{d}{dt}\|f(t)\|^2_{k,N}\bigg|
\leq Ce^{-\gamma t}\|f(t)\|^2_{k+\frac12,N}+Ce^{-\gamma t}\|f(t)\|^2_{k,N}+Ce^{-\frac32\gamma t}\|f(t)\|^3_{k,N},
\end{align*}
and the right hand side is bounded. This shows that the solution is differentiable with respect to time, and this completes the proof. 
\end{proof}

\subsection{Proof of the main theorem}\label{Sec EB}
The main theorem is proved by combining the Propositions \ref{Prop E} and \ref{Prop B} with a suitable iteration. Let $\chi_{ab}(t_0)$, $k_{ab}(t_0)$, and $f(t_0)$ be initial data of the Einstein-Boltzmann system satisfying the conditions of Theorem \ref{Thm}. Let us first define $f_0$ as
\[
f_0(t)=f(t_0).
\]
Since $\|f(t_0)\|_{k+1/2,N}$ is finite with $N\geq 3$, we have $f(t_0,p_*)\leq C\langle p_*\rangle^{-k-1/2}e^{-p^0(t_0)/2}$ for some $C>0$, and this implies that $f_0$ satisfies \eqref{assumption f}:
\[
f_0(t,p_*)\leq C\langle p_*\rangle^{-k}.
\]
Hence, by Proposition \ref{Prop E}, we obtain $\chi_1$ and $k_1$ which are global-in-time solutions to the evolution equations \eqref{evolution1}--\eqref{evolution2} satisfying the properties of Proposition \ref{Prop E}. Note that $\chi_1$ and $k_1$ satisfy the assumptions of Section \ref{Sec AM}: the Hubble variable is initially bounded as $H(t_0)\leq (6/5)^{1/2}\gamma$ and is also non-increasing, so we have $H(t)\leq (6/5)^{1/2}\gamma$ for all $t\geq t_0$, and the estimate \eqref{boundedness F} shows that $F(t)\leq 1/4$ for all $t\geq t_0$. Hence, the $\chi_1$ and $k_1$ satisfy the assumptions of Section \ref{Sec AM}, and we obtain $f_1$ which is a solution of the Boltzmann equation \eqref{Boltzmann} satisfying the properties of Proposition \ref{Prop B}. To be precise, there exists $\varepsilon>0$ such that if $\|f(t_0)\|_{k+1/2,N}<\varepsilon$, then $\|f_1(t)\|_{k,N}\leq C\varepsilon$ for all $t\geq t_0$, which guarantees for $f_1$ the property \eqref{assumption f}:
\[
f_1(t,p_*)\leq C\varepsilon\langle p_*\rangle^{-k}.
\]
Applying Proposition \ref{Prop E} again, we obtain $\chi_2$ and $k_2$ which are solutions to the evolution equations \eqref{evolution1}--\eqref{evolution2} with initial data $\chi_{ab}(t_0)$ and $k_{ab}(t_0)$. In this way, we obtain iterations $\{\chi_n\}_{n=1}^\infty$, $\{k_n\}_{n=1}^\infty$, and $\{f_n\}_{n=1}^\infty$, and consequently solutions to the coupled Einstein-Boltzmann system. The estimate of $\partial_t p^0$ in Lemma \ref{Lem p} and the estimate \eqref{est sigma} show that $p^0$ is bounded globally in time, which ensures the future geodesic completeness (see \cite{Lee} for more details). The asymptotic behaviour of $H$, $\sigma_{ab}\sigma^{ab}$, $\chi_{ab}$, and $\chi^{ab}$ are given by Proposition \ref{Prop E}, and the distribution function is non-negative and satisfies
\[
f(t,p_*)\leq C\varepsilon \langle p_*\rangle^{-k}e^{-\frac12 p^0(t)}
\]
by Proposition \ref{Prop B}. Asymptotic behaviour of the distribution function appears in an orthonormal frame such that
\[
f(t,\p)\leq C\varepsilon (1+e^{2\gamma t}|\p|^2)^{-\frac12 k}e^{-\frac12 p^0},
\]
where $p^0=\sqrt{1+|\p|^2}$ is now independent of $t$. This completes the proof of the main theorem.

\section{Conclusions and outlook}

In this paper we have shown existence and the asymptotic behaviour of forever expanding homogeneous solutions of the Einstein-Boltzmann system with a positive cosmological constant. This extends the work of \cite{Lee} where the Vlasov case was treated to the case of Boltzmann. Note however that the energy method used works well only in the case of Israel particles. It will be left as a future project to consider other types of scattering kernels. We considered spatially homogeneous spacetimes of all Bianchi types except IX, since the latter does not expand forever but recollapses. However it should be possible to obtain results for the case of Bianchi IX before the recollapse following the work of \cite{R} in the Vlasov case (cf. Chapter 26.3 of \cite{R}). Our results rely on the assumption of homogeneity. As a future project we would like to consider general solutions following the work of \cite{R}. A different approach might be to reformulate the system as a symmetric hyperbolic one and follow \cite{Ol}. Another natural generalisation would be to consider the case of a vanishing cosmological constant. There has been recent progress on that matter \cite{Fajman,R2} which in combination with our previous work \cite{LN1} could lead to results in that direction. Finally the well-posedness of solutions to the Einstein-Boltzmann system with an isotropic singularity (see \cite{Tod} for the Vlasov case) is still open and seems an interesting and challenging problem.

\section*{Acknowledgements}
This research was supported by Basic Science Research Program through the National Research Foundation of Korea (NRF) funded by the Ministry of Science, ICT \& Future Planning (NRF-2015R1C1A1A01055216). E.N. has been funded by a Juan de la Cierva research fellowship from the Spanish government and this work has been partially supported by ICMAT Severo Ochoa project SEV-2015-0554 (MINECO).  The authors wish to thank the organisers H\r{a}kan Andreasson, David Fajman and J\'{e}r\'{e}mie Joudioux of the \emph{ESI workshop on the geometric transport equations in General Relativity} for their kind invitation to present our previous work and where we had the opportunity to discuss questions related to the presented work.

\end{document}